\documentclass[11pt]{article} 

\usepackage[utf8]{inputenc} 

\usepackage{geometry} 
\usepackage{graphicx} 
\usepackage{color}

\definecolor{Red}{rgb}{1,0,0}
\definecolor{Blue}{rgb}{0,0,1}
\definecolor{Olive}{rgb}{0.41,0.55,0.13}
\definecolor{Green}{rgb}{0,1,0}
\definecolor{MGreen}{rgb}{0,0.8,0}
\definecolor{DGreen}{rgb}{0,0.55,0}
\definecolor{Yellow}{rgb}{1,1,0}
\definecolor{Cyan}{rgb}{0,1,1}
\definecolor{Magenta}{rgb}{1,0,1}
\definecolor{Orange}{rgb}{1,.5,0}
\definecolor{Violet}{rgb}{.5,0,.5}
\definecolor{Purple}{rgb}{.75,0,.25}
\definecolor{Brown}{rgb}{.75,.5,.25}
\definecolor{Grey}{rgb}{.5,.5,.5}

\usepackage{booktabs} 
\usepackage{array} 
\usepackage{paralist} 
\usepackage{verbatim} 
\usepackage{subfigure} 
\usepackage{amsmath,amssymb,amsthm,mathrsfs}
\usepackage[colorlinks]{hyperref}
\usepackage{blkarray}
\usepackage{bbm}

\usepackage{tikz}  
\usetikzlibrary{cd}

\usepackage{fancyhdr} 
\theoremstyle{plain}

\newtheorem{theorem}{Theorem}

\newtheorem{corollary}{Corollary}[section]

\newtheorem{lemma}{Lemma}[section]

\newtheorem{theorem*}{Theorem}   
\newtheorem{lemma*}{Lemma} 
\newtheorem{corollary*}{Corollary} 
\newtheorem*{remark*}{Remark}

\newtheorem{remark}{Remark}[section]

\makeatletter
\newtheorem*{rep@theorem}{\rep@title}
\newcommand{\newreptheorem}[2]{%
\newenvironment{rep#1}[1]{%
 \def\rep@title{#2 \ref{##1}}%
 \begin{rep@theorem}}%
 {\end{rep@theorem}}}
\makeatother
\newreptheorem{theorem}{Theorem} 
\newreptheorem{lemma}{Lemma}      
\newreptheorem{corollary}{Corollary} 

\newlength{\widebarargwidth}
\newlength{\widebarargheight}
\newlength{\widebarargdepth}

\theoremstyle{definition}
\newtheorem{definition}{Definition}

\def\cA{{\cal A}}

\def\cN{{\cal N}}

\newcommand{\real}{\ensuremath{\mathbb{R}}}

\newcommand{\E}{\ensuremath{\mathbb{E}}}

\begin{document}

\begin{center}

	{\bf{\LARGE{An entropy inequality for symmetric random variables}}}

	\vspace*{.25in}

	\begin{tabular}{ccc}
		{\large{Jing Hao}} & \hspace*{.75in} & {\large{Varun Jog}} \\ 
		{\large{\texttt{jing.hao@wisc.edu}}} & & {\large{\texttt{vjog@wisc.edu}}} \vspace{.2in}
		\\
		Departments of Mathematics & \hspace{.2in} & Department of Electrical \& Computer Engineering \\
		University of Wisconsin - Madison && University of Wisconsin - Madison \\ 
	\end{tabular}

	\vspace*{.2in}

	January 2018

	\vspace*{.2in}

\end{center}

\abstract{We establish a lower bound on the entropy of weighted sums of (possibly dependent) random variables $(X_1, X_2, \dots, X_n)$ possessing a symmetric joint distribution. Our lower bound is in terms of the joint entropy of $(X_1, X_2, \dots, X_n)$. We show that for $n \geq 3$, the lower bound is tight if and only if $X_i$'s are i.i.d.\ Gaussian random variables. For $n=2$ there are numerous other cases of equality apart from i.i.d.\ Gaussians, which we completely characterize. Going beyond sums, we also present an inequality for certain linear transformations of $(X_1, \dots, X_n)$. Our primary technical contribution lies in the analysis of the equality cases, and our approach relies on the geometry and the symmetry of the problem.}

\section{Introduction}
The Entropy Power Inequality (EPI), first proposed by Shannon \cite{Sha48}, states that for any two independent $\mathbb R^n$-valued random variables $X$ and $Y$, 
\begin{equation}\label{eqn: epi}
e^{\frac{2h(X)}{n}} + e^{\frac{2h(Y)}{n}} \leq e^{\frac{2h(X+Y)}{n}},
\end{equation}
where $h(X)$ and $h(Y)$ are the differential entropies of $X$ and $Y$ respectively. Equality holds in inequality \eqref{eqn: epi} if and only if $X$ and $Y$ are Gaussian random vectors with proportional covariance matrices. An equivalent form of  inequality \eqref{eqn: epi} due to Lieb \cite{Lie02}  is also commonly used in the literature, and is stated as follows:
\begin{equation}\label{eq: lieb}
h(\sqrt \lambda X + \sqrt{1-\lambda} Y) \geq \lambda h(X) + (1-\lambda)h(Y),
\end{equation}
where $\lambda \in [0,1]$. Here, equality holds if and only if $X$ and $Y$ are Gaussians with identical covariance matrices. The EPI may be interpreted as a sharp lower bound on the entropy of sums of independent random variables in terms of their individual entropies. It has been widely used in communication theory to prove converses of coding theorems for different kinds of Gaussian channels, such as broadcast channels, wiretap channels, MIMO, etc. \cite{Ber73, LeuHel78, Oza80, Ooh98, WeiEtAl06}.

The EPI was first proved in Stam \cite{Sta59}, and this proof was later simplified in Blachman \cite{Bla65}. A variety of different proofs of the EPI have been discovered since, and we refer the reader to Rioul \cite{Rio11} for an informative and in-depth analysis of different proof strategies.  
Numerous generalizations of this inequality have been proposed over the years such as Costa's inequality for when one of summands is Gaussian \cite{Cos85, LiuEtAl10, CouEtAl17}, a generalization involving subsets random variables in Madiman and Barron \cite{MadBar07}, and a strengthened EPI using an auxiliary random variable in Courtade \cite{Cou16a}.  


In addition to communication theory, the EPI  has also found applications in probability theory for proving the central limit theorem \cite{Joh04book}. Barron \cite{Bar84} established an entropic version of the central limit theorem 
and  conjectured a certain monotonicity property of entropy, which states that entropy is monotonically increasing with respect to the number of summands in the central limit theorem:
\begin{equation}
h\left(\frac{\sum_{i=1}^n X_i}{\sqrt n} \right) \leq h\left(\frac{\sum_{i=1}^{n+1} X_i}{\sqrt{n+1}} \right).
\end{equation}
This conjecture was established in Arstein et al.~\cite{ArtEtAl04}, and simplified proofs were obtained in Madiman and Barron \cite{MadBar07}. More recently, a remarkably short and simple proof was also discovered in Courtade \cite{Cou16a}. Our work in this paper is partly motivated by a series of interesting questions and conjectures made in Ball et al. \cite{BallEtAl16} and Eskenazis et al. \cite{EskEtAl16}, which were in turn motivated by the monotonicity properties of entropy. We briefly describe the work in these papers concerning \emph{directional entropies}.

The monotonicity property of entropy may be interpreted as a result comparing \emph{directional entropies}; i.e., entropy of a random vector projected in a certain direction. Indeed, Eskenazis et al. \cite{EskEtAl16}  interpreted the monotonicity property of entropy as follows: For i.i.d.\ random variables $X_1, \dots, X_n$, entropy along the direction $(1, 1, \dots, 1)^T \frac{1}{\sqrt n}$ is larger than the entropy along $(1, 1, \dots, 1, 0)^T \frac{1}{\sqrt{n-1}}$. This led them to the natural question (Question 6 in \cite{EskEtAl16}): Along which direction is the entropy maximized, or equivalently, which direction is most Gaussian-like for the joint distribution of $(X_1, \dots, X_n)^T$? A natural guess would be that the optimal direction is $(1, 1, \dots, 1)^T \frac{1}{\sqrt n}$, however, this conjecture is not true and it follows from a counterexample constructed in Ball et al. \cite{BallEtAl16} for the case of $n=2$. Ball et al. \cite{BallEtAl16}  conjectured that for log-concave random variables, the entropy maximizing direction for $n=2$ should be $(1,1)^T/\sqrt 2$. In fact, the conjecture in Ball et al. \cite{BallEtAl16} is stronger---$h(\sqrt \lambda X_1 + \sqrt {1-\lambda} X_2)$ is a concave function of $\lambda$. Eskenazis et al. \cite{EskEtAl16} were able to prove (Question 6 in \cite{EskEtAl16}) for a special class of symmetric random variables called Gaussian mixtures. However, in general (Question 6, \cite{EskEtAl16}) and even its special case of  log-concave random variables for $n=2$ is as yet open.  

One of our contributions in this paper is to establish lower bounds on such directional entropies, for \emph{symmetric random vectors}. We call $X$ a symmetric random vector if $f_X(x) = f_X(|x|)$, where by $|x|$ we mean taking the absolute value of each coordinate in the $x$ vector. Stated informally, our result is the following:
\begin{reptheorem}{thm: main}
For a symmetric random vector $X = (X_1, \dots, X_n)^T$ and the unit vector $a = (1, \dots, 1)^T\sqrt n$, the following bound holds:

\begin{align*}
		h( a \cdot  X) = h \left( \frac{\sum_{i=1}^n  X_i}{\sqrt n}  \right) \geq  \frac{h( X)}{n} 
	\end{align*}
	\end{reptheorem}
For an arbitrary unit vector $a = (a_1, \dots, a_n)$, we may use the above result to obtain the bound 
\begin{align*}
h(a \cdot X) \geq \frac{h( X)}{n} + \log \left(n^{n/2}\prod_{i=1}^n a_i\right).
\end{align*}

Notice that unlike most entropy inequalities, our lower bound is in terms of the joint entropy $h(X)$. To our knowledge, this is the first inequality in which entropies of sums and joint entropies both make an appearance. We note that similarities between entropy inequalities for joint distributions \cite{MadTet10} and for sums  \cite{MadBar07} have been observed and explored recently \cite{MadGha17}. Interestingly, our lower bound is maximized in the direction $(1, \dots, 1) \frac{1}{\sqrt n}$, which is the conjectured direction of maximum entropy from \cite{BallEtAl16, EskEtAl16}  for log-concave random vectors.

Another notable point is that our inequality does not require $X_i$'s to be independent. If $X_i$'s are independent, Theorem \ref{thm: main} may be seen as a consequence of the regular EPI stated in Lieb's form \eqref{eq: lieb}. Another interpretation of our bound follows from the following observation: If $A$ is an orthonormal matrix with rows $a_i^T$,  then
\begin{align}\label{eq: ai}
\frac{\sum_{i=1}^n h(a^i \cdot X)}{n} \geq \frac{h(AX)}{n} = \frac{h(X)}{n}.
\end{align}
Inequality \eqref{eq: ai} implies that for any choice of an orthonormal basis, and for any random vector $X$, there is at least one  direction $a^i$ which satisfies $h(a^i \cdot X) \geq \frac{h(X)}{n}.$ Theorem \ref{thm: main} explicitly identifies $(1, \dots, 1)^T/\sqrt n$ as such a direction.

The search for EPI-like inequalities for \emph{dependent} random variables is an active area of research, although there are relatively fewer results here. Takano \cite{Tak95} derived conditions on the joint distribution of $(X,Y)$ such that the classical form of the EPI in \eqref{eqn: epi} continues to hold. Under weaker conditions on the joint distribution compared to those in Takano \cite{Tak95}, Johnson \cite{Joh04} showed that inequality \eqref{eqn: epi} continues to hold with a slight modification: $h(X)$ and $h(Y)$ are replaced by $h(X|Y)$ and $h(Y|X)$ respectively. Rioul \cite{Rio11} showed that  under stronger conditions compared to those of Takano \cite{Tak95}, the equivalent form of the EPI in inequality \eqref{eq: lieb} can be made to hold for any choice of $\lambda \in [0,1]$. 

The conditions from \cite{Tak95} and \cite{Joh04} are in terms of the score functions and Fisher informations of Gaussian perturbed random variables $X_t$ and $Y_t$, where
\begin{align*}
X_t &= X + \sqrt{f_1(t)} Z_1, \text{ and }\\
Y_t &= Y + \sqrt{f_2(t)}Z_2,
\end{align*}
where $Z_1, Z_2 \sim \cN(0,1)$ and $f_1(\cdot), f_2(\cdot)$ are positive functions that tend to $+\infty$ as $t \to +\infty$. As such, these conditions are not easily interpretable. The conditions from Rioul \cite{Rio11} are in terms of the mutual information between Gaussian perturbed versions of $X$ and $Y$ and are more interpretable, although they are not easy to verify. In contrast, we impose the easy to interpret and easy to verify (albeit strong) condition of symmetry on the distribution of $X$.
 
Going beyond directional entropies, it is natural to think about entropies of projections in subspaces of arbitrary dimensions; i.e., considering $h(AX)$ where $A$ is a $k \times n$ matrix with orthonormal rows. Analogous to (Question 6 in \cite{EskEtAl16}), one may pose the question of identifying the entropy maximizing $k$-dimensional subspace for a random vector $X$. This appears to be a very interesting and challenging problem even for specific case of Gaussian mixtures considered in Eskenazis et al. \cite{EskEtAl16}. If $X$ has independent components, Zamir and Feder's EPI for linear transformations of $X$ \cite{ZamFed93} yields lower bounds on $h(AX)$ for a matrix $A$. However, if $X$ has dependent components, no such result is known in the literature. \footnote{The second author learnt of this open problem, as well as the conjectures in Ball et al. \cite{BallEtAl16} and Eskenazis et al. \cite{EskEtAl16} through the Entropy Power Inequalities Workshop sponsored by the American Institute of Mathematics (AIM) held in May, 2017, in San Jose, CA. A list of open problems may be found at this link: \url{http://aimpl.org/entropypower/2/}} We take a step towards obtaining such a result by showing that for symmetric random vectors, it is possible to obtain lower bounds on $h(AX)$ similar to those in Theorem \ref{thm: main}. Stated informally, our result is the following:

\begin{reptheorem}{thm: kdim}
Let $A$ be an orthogonal $k \times n$ matrix with columns $a^1, \dots, a^n$, and suppose that $A$ is ``balanced"; i.e., $\|a^i \|_2= k/n $ for $1 \leq i \leq n$. Then for a symmetric random vector $X$, the following bound holds:
\begin{equation}
h(AX) \geq \frac{k}{n} h(X).
\end{equation}
\end{reptheorem}

\paragraph{Equality cases:} 
Equality holds in the classical EPI if and only if the two random vectors are Gaussian with proportional covariance matrices. We examine conditions for equality in Theorem \ref{thm: main}, and notice an interesting phenomenon where equality conditions change when dimension increases beyond $n=2$. Intuitively, this is because the symmetry assumption becomes stronger in higher dimensions. In Theorem \ref{thm: equality2}, we completely characterize all equality cases for $n=2$ by showing that $X$ has to be a $45^\circ$ rotated version of $Z$, where $Z$ has i.i.d.\ and symmetric components:
\begin{reptheorem} {thm: equality2}
If $h((X_1+X_2)/\sqrt 2) = h(X_1,X_2)/2$, then there exist i.i.d.\, symmetric random variables $Z_1$ and $Z_2$ such that 
\begin{align*}
\frac{1}{\sqrt 2}\begin{pmatrix} 1 & -1\\ 1 & 1 \end{pmatrix} \begin{pmatrix} Z_1 \\ Z_2 \end{pmatrix} = \begin{pmatrix} X_1\\X_2 \end{pmatrix}.
\end{align*}
\end{reptheorem}
For dimensions $n >2$, we show that the only condition under which equality holds in Theorem \ref{thm: main} is  when $X$ has i.i.d.\ Gaussian components. This result is the most technical part of our paper and is given by Theorem \ref{thm: equality}:

\begin{reptheorem}{thm: equality}
Let $X$ be a symmetric random vector in $\mathbb R^n$, where $n > 2$. Equality holds in Theorem \ref{thm: main} for the direction $(1, \dots, 1)^T/\sqrt n$ if and only if $X_i$ are i.i.d.\ Gaussian random variables.
\end{reptheorem}
It is interesting to note that the symmetry condition combined with the equality condition forces $X$ to have independent components, even when we allow for dependence.

The structure of our paper is as follows: The main inequality and its proof will be given in Section \ref{section: main}.  In Section \ref{sec:the_equality_cases}, we completely characterize the equality cases for $n=2$ and $n > 2$ separately.  In Section \ref{section: extensions}, we extend the inequality to $k-$dimensional projections.  All supporting lemmas used in the proof of the main results will be given in Appendix \ref{appendix: a}.


\paragraph{Notation:} 
We denote a random column vector in $\mathbb R^n$ by $X = (X_1, X_2, \dots, X_n)^T$. We will use the notation $x = (x_1, x_2, \dots, x_n)^T$ and denote the joint density $f_{X_1, \dots, X_n}(x_1, \dots, x_n)$ by $f_X(x)$. For an $\mathbb R^n$-valued random variable $X$ with a continuous and differentiable density $f_X$, the differential entropy of $X$ is given by
\begin{align*}
	h(X) = -\int_{\mathbb R^n} f_X(x) \log f_X(x) dx
\end{align*}
and its Fisher information matrix is defined as 
\begin{align*}
	[I(X)] =  \E \left[ (\nabla \log f_X(x)) (\nabla \log f_X(x))^T \right].
\end{align*}
Note that the $(i,j)$-th term is given by
\begin{align*}
	[I(X)]_{i,j} &= \E \left[\frac{\frac{\partial f_X(x)}{\partial x_i}}{f_X(x)} \cdot \frac{\frac{\partial f_X(x)}{\partial x_j}}{f_X(x)}\right]\\
	&= \int_{\real^n} \frac{\partial f_X(x)}{\partial x_i}  \frac{\partial f_X(x)}{\partial x_j} \frac{1}{f_X(x)} dx.
\end{align*}
The Fisher information of $X$ is defined as 
\begin{align*}
	I( X) &= \text{Trace}([I(X)])\\
	&= \E \left[\sum_{i=1}^n \left(\frac{\frac{\partial f_X(x)}{\partial x_i}}{f_X(x)}\right)^2\right]\\
	&=  \int_{\mathbb R^n}  \frac{\| \nabla f_X(x) \|^2_2}{f_X(x)} dx.
\end{align*}
The quantity $\rho_X(x) := \nabla \log f_X(x)$ is called the score function of $X$, and the Fisher information is also stated as $\E[\|\rho(X)\|_2^2]$.

\section{Main results}
\label{section: main}
\begin{definition}
	A random vector  $X = (X_1, X_2, \dots, X_n)^T$ on $\mathbb R^n$ with a density $f_{ X}$ is called a symmetric random vector if the following holds:
	\begin{align}
		f_{ X}(x_1, x_2, \dots, x_i, \dots, x_n) = f_{ X}(x_1, x_2, \dots, -x_i, \dots, x_n),
	\end{align}
	for all $1 \leq i \leq n$ and all $x_i \in \mathbb R$. An equivalent way to state this is $f_X(x) = f_X(|x|)$, where $|x| = (|x_1|, \dots, |x_n|).$
\end{definition}
\begin{theorem}\label{thm: main}
	Let $ X = (X_1, X_2, \dots, X_n)^T$ be a symmetric random vector on $\real^n$. Then the following inequality holds:
	\begin{align}
		h \left(\sum_{i=1}^n \frac{X_i}{\sqrt n} \right) \geq \frac{h(X)}{n} &= \frac{h(X_1, \dots, X_n)}{n}.
		\label{eqn: main}
	\end{align}
\end{theorem}
 Following a well-established strategy for proving such inequalities, we first prove a Fisher information inequality and then use an integral representation of differential entropy to arrive at the desired inequality. Lemma \ref{lemma: fisher} gives the Fisher information inequality that we need.
\begin{proof}[Proof of Theorem \ref{thm: main}]
	Let $Y = \frac{\sum_{i=1}^n X_i}{\sqrt n}$. From Lemma \ref{lemma: fisher}, we have $I(Y) \leq I(X)/n$. Consider a random Gaussian vector $Z$ with the identity covariance matrix; i.e.,  $Z \sim \cN(0, I)$, and $Z$ is independent of $X$. For $t \geq 0$, define $X_t := X + \sqrt t Z$. Note that 
	\begin{align*}
		\frac{\sum_{i=1}^n X_t(i)}{\sqrt n} &= \frac{ \sum_{i=1}^n X_i}{\sqrt n} + \sum_{i=1}^n \frac{\sqrt t Z_i}{\sqrt n} \stackrel{d} = Y + \sqrt t Z_0 := Y_t,
	\end{align*}
	where $\stackrel{d}=$ stands for equality in distribution, and $Z_0 \sim \cN(0,1)$ is independent of $Y$. Note also that if $X$ has a symmetric joint distribution, $X_t$ also has a symmetric distribution. Hence, we may apply Lemma~\ref{lemma: fisher} to $X_t$ and $Y_t$ to conclude
	\begin{align}
		I(Y_t) \leq \frac{I(X_t)}{n}.
	\end{align}
	We now use the integral form of differential entropy in terms of Fisher information \cite{Rio11}, which implies
	\begin{align}
		h(X) &= \frac{n}{2} \log 2\pi e - \frac{1}{2} \int_{0}^\infty \left(I(X_t) - \frac{n}{1+t} \right)dt, \text{ and }\\
		h(Y) &= \frac{1}{2} \log 2 \pi e - \frac{1}{2} \int_0^\infty \left(I(Y_t) - \frac{1}{1+t} \right)dt.
	\end{align}
	This implies
	\begin{align*}
		\frac{h(X)}{n} &=  \frac{1}{2} \log 2\pi e - \frac{1}{2} \int_{0}^\infty \left(\frac{I(X_t)}{n} - \frac{1}{1+t} \right)dt\\
		&\leq  \frac{1}{2} \log 2\pi e - \frac{1}{2} \int_{0}^\infty \left(I(Y_t) - \frac{1}{1+t} \right)dt\\
		&= h(Y).
	\end{align*}
	This completes the proof.
\end{proof}

The following corollary is an immediate consequence of Theorem \ref{thm: main} and the scaling properties of the entropy function:
\begin{corollary}
	Let $ a = (a_1, \dots, a_n)^T$ be any unit vector in $\mathbb R^n$; i.e., $\| a \|_2 =1.$ Then the following inequality holds:
	\begin{align*}
		h( a \cdot  X) = h \left( \sum_{i=1}^n a_i X_i \right) \geq  \frac{h( X)}{n} + \log \left(n^{n/2}\prod_{i=1}^n a_i\right).
	\end{align*}
\end{corollary}

\section{Equality cases}
\label{sec:the_equality_cases}

Let $v^1 = \left( \frac{1}{\sqrt{n}}, \cdots, \frac{1}{\sqrt{n}} \right)^T$. Extend $v^1$ to an orthogonal basis $\{v^1,\cdots, v^n\}$.  Denote 
\begin{equation*}
	A=( v^1, \cdots , v^n)
\end{equation*}
Let $Z=(Z_1, \cdots, Z_n)^T:= A^TX$, then $X= A Z$.  In particular, $Z_1 = \frac{1}{\sqrt{n}} \sum^{n}_{i=1} X_i$.  Let $z = A^T x$, we have 
\[
	f_Z(z) = f_X(x) \cdot \Big|\det\left( \frac{dx}{dz}\right)\Big| = f_X(x) \cdot |A| = f_X(x),
\]
and 
\begin{equation*}
	\frac{\partial f_Z(z)}{z_1} = \sum^{n}_{i=1} \frac{1}{\sqrt{n}} \frac{\partial f_X(x)}{\partial x_i}
\end{equation*}
Denote $ \widehat{Z_i} = (Z_1, \cdots, Z_{i-1}, Z_{i+1}, Z_n)$.  We will show $Z_1$ is independent with $\widehat{Z_1}$ using Lemma \ref{lem: general indep}, which is a general result on the independence of random variables.

\begin{lemma}\label{lemma: z1}
	For equality to hold in inequality (\ref{eqn: main}), $Z_1$ must be independent of $\widehat{Z_1}$.	
	\label{lem: Z_1 indep}
\end{lemma}

\begin{proof}[Proof of Lemma \ref{lemma: z1}]
For equality to hold in Theorem \ref{thm: main}, the only condition we need is that inequality \eqref{eqn_var_mean} from Lemma \ref{lemma: fisher} is an equality for all $y$. The only way this can happen is if $$\frac{ \sum_{i=1}^n \frac{\partial f_X(x)}{\partial x_i}}{f_X(x)}$$ 
is constant for all $x$ on the hyperplane $x \cdot (1, \dots, 1)^T/\sqrt n = y$. An equivalent way to express this is  $ \frac{\partial f_Z(z)}{\partial z_1}/f_Z(z)$ is a function of $z_1$.  This is equivalent to 
	\begin{equation*}
		\frac{\partial^2 }{\partial z_k \partial z_1} \log f_Z(z) = 0, \: \forall k \neq 1.
	\end{equation*}
	Applying Lemma \ref{lem: general indep}, we conclude that $Z_1$ is independent of $\widehat{Z_1}$.
\end{proof}

The symmetry assumption may be combined with Lemma \ref{lemma: z1} to yield a stronger independence result:
\begin{lemma}\label{lemma: A}
Let $X$ be a symmetric random variable that achieves equality in Theorem \ref{thm: main}. Let $\cA = \{x \in \mathbb R^n \text{ such that } |x| = (1, \dots, 1)^T/\sqrt n\}$. Let $ B= (b^1, \dots, b^n)$ be any orthogonal basis such that $b^1 \in \cA$. If $Y = B^T X$, then $Y_1$ is independent of $(Y_2, \dots, Y_n)$.
\end{lemma}
\begin{proof}[Proof of Lemma \ref{lemma: A}]
Define
$$S = \{1 \leq i \leq n | b^1(i) = -1/\sqrt n\}.$$ where $b^1(i	)$ is the $i^{th}$ coordinate of $b^i$.  For $i \in S$, define $\widetilde X_i = -X_i$, and otherwise $\widetilde X_i = X_i$. Define $\widetilde B$ to be equal to the $B$ matrix, except the $i$-th row is flipped in sign if $i \in S$. Note that $\widetilde B$ is now an orthogonal matrix with its first column being $(1, \dots, 1)^T/\sqrt n$. Clearly, $Y = \widetilde B^T \widetilde X$. Using symmetry, it is clear that if $X$ satisfies Theorem \ref{thm: main} with equality, then so does $\widetilde X$. Applying Lemma \ref{lemma: z1}, we conclude $Y_1$ is independent of $(Y_2, \dots, Y_n)$.
\end{proof}


\subsection{Equality conditions for $n=2$}
\label{sub:when_n_2_}
When $n=2$, the only orthonormal basis with one vector being $(1,1)^T/\sqrt 2$ is 
$$A = \frac{1}{\sqrt 2}
\begin{pmatrix}
1 &-1\\
1 & 1
\end{pmatrix}.
$$
Define $Z = A^T X$. We first show that $Z_1$ and $Z_2$ are independent and identically distributed:
\begin{lemma}\label{lemma: z1z2}
	If $X$ is a symmetric random vector in $\mathbb R^2$ such that inequality \eqref{eqn: main} from Theorem \ref{thm: main} is satisfied with equality, then $Z_1$ and $Z_2$ are independent, identically distributed, symmetric random variables.
		\label{lem: n=2 i.i.d.}
\end{lemma}
\begin{proof}[Proof of Lemma \ref{lemma: z1z2}]
By Lemma \ref{lemma: z1}, we obtain the independence of $Z_1$ and $Z_2$. Since $X$ is symmetric, we see that $X_1+X_2$, $X_1-X_2$, $-X_1+X_2$, and $-X_1-X_2$ have the same distributions. This means that $Z_1, Z_2, -Z_1$ and $-Z_2$ have identical distributions, and this concludes the proof. 
\end{proof} 


\begin{theorem}\label{thm: equality2}
	For a symmetric random vector $X$ in $\mathbb R^2$, equality holds in inequality (\ref{eqn: main}) for Theorem \ref{thm: main} if and only if $Z_1 = \frac{X_1 + X_2}{\sqrt{2}}$ and $ Z_2 = \frac{X_1 - X_2}{\sqrt{2}}$ are independent, identically distributed symmetric random variables.
\end{theorem}

\begin{proof}
	The ``only if" part was established in Lemma \ref{lemma: z1z2}, so we need to check the ``if" part. If $Z_1 = \frac{X_1 + X_2}{\sqrt{2}}$ and $Z_2 = \frac{X_1 - X_2}{\sqrt{2}}$ are i.i.d. according to a symmetric distribution $f_Z(\cdot)$, then 
	\begin{equation*}
		f_X(x_1, x_2) = f_{Z}\left( \frac{x_1 + x_2}{\sqrt{2}}\right) f_{Z}\left( \frac{x_1- x_2}{\sqrt{2}}\right).
	\end{equation*}
	It is easy to check that $X$ is indeed a symmetric random variable. Furthermore, we may also check that inequality \eqref{eqn: main} from Theorem \ref{thm: main} holds with equality:
	\begin{align*}
		\frac{h(X_1,X_2)}{2} = \frac{h(Z_1, Z_2)}{2} = \frac{h(Z_1) + h(Z_2)}{2} = h(Z_1) = h\left( \frac{X_1 + X_2}{\sqrt{2}}\right).
	\end{align*}
	This completes the proof.
\end{proof}

\subsection{Equality conditions for $n \geq 3$}

\begin{theorem}\label{thm: equality}
Let $X = (X_1, \dots, X_n)^T$ be a symmetric $\mathbb R^n$-valued random vector. Then equality holds in Theorem \ref{thm: main} if and only if $X_i$'s are i.i.d.\ 0-mean Gaussian random variables.
\end{theorem}
\begin{proof}
Let $\cA$ be the set of all unit vectors where each coordinate is $\frac{\pm 1}{\sqrt n}$, that is
\begin{align*}
\cA = \{x \in \mathbb R^n \text{ such that } |x| = (1, \dots, 1)^T/\sqrt n\}.
\end{align*}
From $\cA$, choose any $n$ vectors $v^1, v^2, \dots, v^n$ such that $\text{span}(v_{1:n}) = \mathbb R^n$. Here, the notation $v_{1:n}$ is shorthand for the set $\{v^1, \dots, v^n\}$. Furthermore, choose $v^1$ and $v^2$ such that $v^1 \not \perp v^2$. Note that it is always possible to make such a choice for $n>2$, but not for $n=2$.  For example, we can let 
 \[
 v^i(j)= \begin{cases}
 1/\sqrt{n} & \text{if } i \neq j \\
 -1/\sqrt{n} & \text{if } i = j
 \end{cases}
 \]
 where $v^i(j)$ is the $j^{th}$ coordinate of $v^i$.  Let $M$ be the matrix with the $i^{th}$ column as $v^i$, then $M = (N - 2I)/\sqrt{n}$ where $N$ is the $n$ by $n$ matrix with all entries equal to $1$.  Clearly, $N$ satisfy
 \[
 N^2 = n N
 \]
and
 \[
 M (N + (2-n)I) = \frac{1}{\sqrt{n}} (N - 2I)(N + (2-n)I) = \frac{-4+2n}{\sqrt{n}}I
 \]
 Since $n \neq 2$, $M$ is invertible, and so $v_1, \cdots, v_n$ are linearly independent. $v_1 \cdot v_2 = \frac{n-4}{n}$, therefore $v_1 \not\perp v_2$ when $n\neq4$.  For $n=4$, change $v^1$ to $\frac{1}{\sqrt{n}}[1,1,1,1]$, explicit computations can show they satisfy our conditions.
 
 By Lemma \ref{lemma: A} we know that $X \cdot v^i$ is independent of $\text{Proj}(X, (v^i)^\perp)$. Here, $\text{Proj}(X, (v^i)^\perp)$ stands for the projection of $X$ on to the $(n-1)$-dimensional subspace that is orthogonal to $v^i$. We first construct an orthonormal basis by applying the Gram-Schmidt procedure to $v^1, \dots, v^n$ in that order. Denote this basis by $A_1 = (a_{11}, a_{12}, \dots, a_{1n}$), where each $a_{1j}$ is a vector in $\mathbb R^n$. Observe that since we used the Gram-Schmidt procedure, we must have  $\text{span}(a_{11:1j}) = \text{span}(v_{1:j})$ for all $j$. In particular, we have $a_{11} = v^1$. 

In a similar manner, for each $i \ge 2 $ denote by $A_i = (a_{i1}, a_{i2}, \dots, a_{in})$ the orthogonal basis obtained by using the Gram-Schmidt procedure on the permutation with $v^i$ used first; i.e., on $v^i, v^1, \dots, v^n$. Note that as for $i=1$, we have that $a_{i1} = v^i$. We also have that the last $n-i$ columns of $A_1$ and $A_i$ are identical, because 
$$\text{span}(a_{11:1i}) = \text{span}(a_{i1:ii}) = \text{span}(v_{1:i}).$$ 
This is illustrated in the below diagram:

\[
\begin{tikzcd}
{[v_1,v_2,\cdots,v_n]} \rar{G-S} & A_1 := [a_{11}=v_1, a_{12}, \cdots,a_{1n}] \ar{d}{\text{carry over last }n-2 }\\
  {[v_2, v_1, v_3, \cdots,v_n]} \rar{G-S} & A_2:=[a_{21}=v_2,  a_{22},a_{23} = a_{13}, \cdots, a_{2n}=a_{1n}]
\ar{d}{\text{carry over last }n-3 }\\
{[v_3, v_1,v_2 ,v_4, \cdots,v_n]} \rar{G-S} & A_3:=[a_{31}=v_3,  a_{32},a_{33},a_{34} = a_{24}, \cdots, a_{3n}=a_{2n}] \\
\vdots & \vdots  \\
{[v_n, v_1, \cdots,v_{n-1}]} \rar{G-S} & A_n:=[a_{n1}=v_n,  a_{n2}, \cdots, a_{nn}]
\end{tikzcd}
\]

We now define new random variables which are  the projections of $X$ on to the basis given by the bases $A_i$'s. For $1 \leq i \leq n$, define
$$Z^i = A_i^T X.$$ Denote $Z^i(j)$ as the $j^{th}$ components of $Z^i$.
Our strategy is to show that the components of $Z^1$ are independent and Gaussian with the same variance. Note that this implies that $X$ is a spherically symmetric Gaussian, and concludes the proof of Theorem \ref{thm: equality}. 
 
Let $R_i$ be a rotation matrix such that $A_i = A_1 R_i$, for $1 \leq i \leq n$. Since the last $n-i$ columns of $A_i$ and $A_1$ are identical, we have that 
\begin{align*}
R_i = \begin{pmatrix}
&\widehat R_i &0_{i \times n-i}\\
&0_{n-i \times i} & I_{n-i \times n-i}
\end{pmatrix}
\end{align*}
for some $i \times i$ rotation matrix $\widehat R_i$. We can express $Z^i$ in terms of $Z^1$ via the relation
 \begin{align*}
 Z^i &= A_i^T X\\
 & = R_i^T A_1^T X\\
 &= R_i^T Z^1.
 \end{align*} 
For $i=2$ in particular, let 
$$\widehat R_2^T = \begin{pmatrix}
\alpha_{11} &\alpha_{12}\\
\alpha_{21} & \alpha_{22}.
\end{pmatrix}.$$
Since we chose $v^1$ and $v^2$ to be non-orthogonal, all entries of $\widehat R_2^T$ are non-zero. Furthermore, we have 
\begin{align*}
\begin{pmatrix}Z^2(1)\\ Z^2(2) \end{pmatrix} = \widehat R_2^T \begin{pmatrix}  Z^1(1)\\ Z^1(2) \end{pmatrix}.
\end{align*}
Additionally, we may apply Lemma \ref{lemma: A}  to write the following independence relations:
\begin{align*}
Z^1(1) &\perp \!\!\! \perp Z^1(2), \quad \text{ and }\\
Z^2(1) &\perp \!\!\! \perp Z^2(2).
\end{align*}
This shows that $(Z^1(1), Z^1(2))$ has independent components even after rotating by a matrix $\widehat R_2$. Since $\widehat R_2^T$ has all  non-zero entries, we may use the characterization theorem of Gaussian distributions \cite{GhuOlk62} to conclude that $Z^1(1)$ and $Z^1(2)$ are normally distributed with 0 mean and identical variances. 

We will continue the proof using induction.  Assume $Z^1(1), \cdots, Z^1(k-1)$ are i.i.d. Gaussian and independent of $Z^1(k), \cdots,Z^1(n)$.  This certainly is true for $k=1$.  We are going to show that $Z^1(1), \cdots, Z^{1}(k)$ are i.i.d. Gaussian as well, and independent of $Z^1(k+1), \dots, Z^1(n)$.  Express $Z^k(1)$ as a linear combination of $Z^1(1), \cdots, Z^{1}(k)$, 
\[
Z^k(1) = \alpha_{11}Z^1(1) + \cdots + \alpha_{1,k}Z^1(k),
\]
for some $\alpha_{11}$ through $\alpha_{1k}$. By Lemma \ref{lemma: A} and our induction assumption, we have the following independence relations:
\begin{align*}
\alpha_{11} Z^1(1) + \cdots + \alpha_{1,k-1} Z^{1}(k-1) &\perp \!\!\! \perp (Z^1(k+1), \dots, Z^1(n)), \quad \text{ and }\\
\alpha_{11} Z^1(1) + \cdots + \alpha_{1,k-1} Z^{1}(k-1) + \alpha_{1,k}Z^{1}(k) &\perp \!\!\! \perp (Z^1(k+1), \dots, Z^1(n))
\end{align*}
Note that $\alpha_{1,k} \neq 0$ because $\alpha_{1,k}$ corresponds to the projection of $X$ onto $v_k$ which does not live in the span of $v_1, \cdots, v_{k-1}$. We may therefore apply Lemma \ref{lemma: independence} to conclude that 
\begin{align*}
Z^1(k) &\perp \!\!\! \perp (Z^1(k+1), \dots, Z^1(n)).
\end{align*}
Furthermore, we also know that $Z^k(1) \sim Z^1(1) = \cN(0, \sigma^2)$ because of the symmetry assumption. We also have 
$$\alpha_{11} Z^1(1) + \cdots + \alpha_{1,k-1} Z^1(k-1) \sim \cN(0, \alpha_{11}^2 \sigma^2 +\cdots + \alpha_{1,k-1}^2 \sigma^2),$$
Note that $ \sum_{i=1}^k \alpha_{1,i}^2 =1$. Applying Lemma \ref{lemma: gaussian}, we conclude that $Z^1(k) \sim \cN(0, \sigma^2)$. By induction, all the $Z^1(i)$'s are independent and identically distributed as Gaussian random variables and completes the proof. \end{proof}

\begin{remark}
Note that we use the Lemma \ref{lemma: independence} only for the case when $X$ is a Gaussian random variable, in which case $\phi_X$ has no roots.
\end{remark}

\section{Extensions}
\label{section: extensions}
The proof of Theorem \ref{thm: main} may be adapted to obtain a version of the same with $k$-dimensional projections:
\begin{theorem}\label{thm: kdim}
Let $X$ be a symmetric $\mathbb R^n$-valued random vector. Let $A$ be a $k \times n$ matrix with orthonormal rows. The matrix $A$ is assumed to be balanced; i.e. all columns of $A$ have the same $\ell_2$-norm:
\begin{align*}
\sum_{i=1}^k a_{ij}^2 = \frac{k}{n} \quad \text{ for all } 1 \leq j \leq n.
\end{align*}
Then the following bound holds for $h(AX)$:
\begin{equation}\label{eq: kdim}
h(AX) \geq \frac{k}{n}h(X).
\end{equation}
\end{theorem}

\begin{proof}
Lemma \ref{lemma: matrix fisher} from Appendix \ref{appendix: a} gives 
\begin{align}
\frac{\nabla f_Y(y)}{f_Y(y)} = \mathbb E \left[ A \left(\frac{\nabla f_X(X)}{f_X(X)}\right)\Big | Y = y\right].
\end{align}
Taking the squared norm on both sides and using Cauchy-Schwartz inequality, we see that
\begin{align}
\| \rho(y) \|_2^2 \leq  \mathbb E \left[ \left \|A \left(\frac{\nabla f_X(X)}{f_X(X)}\right) \right \|_2^2\Big | Y = y\right].
\end{align}
Taking an expectation with respect to $Y$, we conclude
\begin{align}
I(Y) \leq \mathbb E \left[ \left \|A \left(\frac{\nabla f_X(X)}{f_X(X)}\right) \right \|_2^2\right].
\end{align}
Notice that upon expanding, the right hand side contains cross terms of the form 
$$\E \left[\frac{\frac{\partial f_X(X)}{\partial x_i}\cdot \frac{\partial f_X(X)}{\partial x_j}}{f_X(X)^2} \right],$$
which are all equal to 0 due to Lemma \ref{lemma: cross terms} in Appendix \ref{appendix: a}. Furthermore, the coefficient of each term of the form 
$$\E \left[\frac{\left(\frac{\partial f_X(X)}{\partial x_j}\right)^2}{f_X(X)^2} \right] $$
is given by $\sum_{i=1}^k a_{ij}^2$, which is $k/n$ since $A$ is assumed to be a balanced matrix. Thus, we arrive at the bound
\begin{align}
I(Y) \leq \frac{k}{n} I(X).
\end{align}
As in Theorem \ref{thm: main}, we now use the integral form of differential entropy in terms of Fisher information \cite{Rio11}, which implies
	\begin{align}
		h(X) &= \frac{n}{2} \log 2\pi e - \frac{1}{2} \int_{0}^\infty \left(I(X_t) - \frac{n}{1+t} \right)dt, \text{ and }\\
		h(Y) &= \frac{k}{2} \log 2 \pi e - \frac{1}{2} \int_0^\infty \left(I(Y_t) - \frac{k}{1+t} \right)dt.
	\end{align}
	This implies
	\begin{align*}
		\frac{kh(X)}{n} &=  \frac{k}{2} \log 2\pi e - \frac{1}{2} \int_{0}^\infty \left(\frac{kI(X_t)}{n} - \frac{k}{1+t} \right)dt\\
		&\leq  \frac{k}{2} \log 2\pi e - \frac{1}{2} \int_{0}^\infty \left(I(Y_t) - \frac{k}{1+t} \right)dt\\
		&= h(Y).
	\end{align*}
	This completes the proof.

\end{proof}



\section{Conclusion}

In this paper, we discovered a new lower bound for directional entropies of symmetric random variables. Our bounds are different from similar bounds in the literature in two key aspects. Firstly, the lower bound is in terms of the joint entropy $h(X)$  instead of being in terms of a linear combination of $h(X_i)$'s. And secondly, the lower bound holds for dependent random variables as well, as long as the joint distribution is symmetric. Our proof strategy involves proving a Fisher information inequality and then deriving a corresponding entropy inequality.  

Our main technical contribution is the analysis of the equality cases. For $n=2$, we completely characterized all possible equality cases, and showed that unlike the regular EPI, equality may hold even for non-Gaussian random variables. For dimensions great than 2, we showed that equality holds if and only if the random variables are  i.i.d. Gaussian. Although this is the same equality condition for the regular EPI, our proof techniques are novel and rely on certain independence properties of the joint distribution combined with the symmetry assumption. Lastly, we also proved a generalization that yields entropy bounds for certain projections into $k$-dimensional subspaces.  

There are a number of open problems that we would like to consider in future work. For $k$-dimensional projections, our lower bound only holds for certain ``balanced" projection matrices. It would be interesting to see if these bounds can yield bounds for projections in arbitrary $k$-dimensional subspaces as well. Analyzing equality cases for such $k$-dimensional projections also appears to be a challenging problem. It would be quite surprising if there are equality cases other that i.i.d.\ Gaussian random variables in these problems. Yet another direction to pursue would be to examine joint distributions corresponding to i.i.d.\ Gaussian mixture random variables as in Eskenazis et al. \cite{EskEtAl16}, and identify which $k$-dimensional projections have the largest entropy. The current results in \cite{EskEtAl16} hold only for 1-dimensional projections. Our analysis in this paper heavily relies on the symmetry assumption. It is easy to construct examples of non-symmetric random variables that do not satisfy the bounds in this paper. It would be interesting to see we can establish similar bounds for other classes of joint distributions, such as those corresponding to certain symmetric graphical models.

\begin{appendix}
\section{Proofs of lemmas}\label{appendix: a}
\begin{lemma}\label{lemma: matrix fisher}
Let $X$ be an $\mathbb R^n$-valued random variable. Let $A$ be any $k \times n$ orthogonal matrix, and let $Y = AX$. Then we have the equality
\begin{equation}
\frac{\nabla f_Y(y)}{f_Y(y)} = \mathbb E \left[ A \left(\frac{\nabla f_X(X)}{f_X(X)}\right)\Big | Y = y\right].
\end{equation}
\end{lemma}
\begin{proof}[Proof of Lemma \ref{lemma: matrix fisher}]
Let the rows of $A$ be $A_1^T, \dots, A_k^T$. By finding orthogonal rows $A_{k+1}^T, \dots, A_n^T$, we extend the set of rows of $A$ to a complete orthogonal basis. For a point $y = (y_1, \dots, y_k)^T$, we have that $y_0 = A_1y_1+ \dots + A_k y_k$ is vector satisfying $Ay_0 = y$. Thus, we may write $f_Y(y)$ as
\begin{align}\label{eq: fy}
f_Y(y) = \int_{t} f_X(y_0+A_{k+1}t_1 + \dots + A_n t_{n-k})dt,
\end{align}
where $t = (t_1, \dots, t_{n-k})^T$. Define
$y_t = y_0+A_{k+1}t_1 + \dots + A_n t_{n-k}$. Taking the gradient with respect to $y$ on both sides of equation \eqref{eq: fy}, it is easy to check that
\begin{align}\label{eq: fy2}
\nabla f_Y(y) &= \int_t A \nabla f_X(y_t) dt.
\end{align}
Dividing both sides of equation \eqref{eq: fy2} by $f_Y(y)$, 
\begin{align}
\frac{\nabla f_Y(y)}{f_Y(y)} = \int_t A \frac{\nabla f_X(y_t)}{f_X(y_t)} \cdot \frac{f_X(y_t)}{f_Y(y_t)} dt.
\end{align}
Noting that $f_{X|Y}(y_t | y) = \frac{f_X(y_t)}{f_Y(y_t)}$, we conclude that
\begin{align}
\frac{\nabla f_Y(y)}{f_Y(y)} &= \int_t A \frac{\nabla f_X(y_t)}{f_X(y_t)}f_{X|Y}(y_t | y)  dt\\
&=  \mathbb E \left[ A \left(\frac{\nabla f_X(X)}{f_X(X)}\right)\Big | Y = y\right].
\end{align}
\end{proof}
\begin{lemma}\label{lemma: fisher}
Let $X = (X_1, X_2, \dots, X_n)$ be a symmetric $\mathbb R^n$-valued random vector, and let  $Y = \frac{\sum_{i=1}^n X_i}{\sqrt n}$. Then we have
\begin{align}
I(Y) &\leq \frac{I(X)}{n}.
\end{align}
\end{lemma}

\begin{proof}[Proof of Lemma \ref{lemma: fisher}]
		%

Lemma \ref{lemma: matrix fisher} from Appendix \ref{appendix: a} gives that 
\begin{align}
\frac{\nabla f_Y(y)}{f_Y(y)} = \mathbb E \left[ \frac{1}{\sqrt n}  \frac{\sum_{i=1}^n\frac{\partial f_X(X)}{\partial x_i}}{f_X(x)} \Bigg | Y = y\right].
\end{align}
Using Cauchy-Schwartz inequality, we have
	\begin{equation}
		\rho(y)^2 \le \E \left[ \left(\frac{1}{\sqrt n} \frac{ \sum_{i=1}^n \frac{\partial f_X(X)}{\partial x_i}}{f_X(X)} \right)^2\Bigg | Y = y \right].
		\label{eqn_var_mean}
	\end{equation}
	Averaging with respect to $y$,
	\begin{align}
		I(Y) &= \E[\rho(y)^2] \\
		&\le \E \left[  \left(\frac{1}{\sqrt n} \frac{ \sum_{i=1}^n \frac{\partial f_X(X)}{\partial x_i}}{f_X(X)} \right)^2 \right] \\
		&= \E\left[ \sum^{n}_{i=1} \frac{1}{n} \left(\frac{\frac{\partial f_X(X)}{\partial x_i}}{f_X(X)} \right)^2 + 2 \sum_{ 1 \le i < j \le n} \frac{1}{n} \frac{\frac{\partial f_X(X)}{\partial x_i}\cdot \frac{\partial f_X(X)}{\partial x_j}}{f_X(X)^2} \right]\\
		&= \E\left[ \sum^{n}_{i=1} \frac{1}{n} \left(\frac{\frac{\partial f_X(X)}{\partial x_i}}{f_X(X)} \right)^2\right] + \frac{2}{n} \E \left[\sum_{ 1 \le i < j \le n}\frac{\frac{\partial f_X(X)}{\partial x_i}\cdot \frac{\partial f_X(X)}{\partial x_j}}{f_X(X)^2} \right]\\
		&= \frac{I(X)}{n} + \frac{2}{n} \E \left[\sum_{ 1 \le i < j \le n}\frac{\frac{\partial f_X(X)}{\partial x_i}\cdot \frac{\partial f_X(X)}{\partial x_j}}{f_X(X)^2} \right]. \label{eqn_sec_term}
	\end{align}
Lemma \ref{lemma: cross terms} from Appendix \ref{appendix: a} shows that since $X$ is symmetric, each term of the form 
$$\E \left[\frac{\frac{\partial f_X(X)}{\partial x_i}\cdot \frac{\partial f_X(X)}{\partial x_j}}{f_X(X)^2} \right] = 0,$$
and thus the second term in expression \eqref{eqn_sec_term} vanishes, and this concludes the proof.
\end{proof}
\begin{lemma}
	Let $X=(X_1, \cdots, X_n)$ be a random vector and $f_X$ is its density function.  Define
	\begin{equation*}
		\widehat{X_i} := (X_1, \cdots, X_{i-1}, X_{i+1}, \cdots, X_n).
	\end{equation*}
	Then $X_i$ is independent of $\widehat{X_i}$ if and only if 
	\begin{equation*}
		\frac{\partial ^2}{\partial x_k \partial x_i} \log f_X(x) = 0, \: \forall k \neq i
	\end{equation*}
	\label{lem: general indep}
\end{lemma}

\begin{proof}
	If $X_i$ is independent of $\widehat{X_i}$, then 
	\begin{equation*}
		f_X(x) = f_{X_i}(x_i) f_{\widehat{X_i}}(\widehat{x_i}).
	\end{equation*}
	Taking the logarithm and differentiating with respect to $x_i$, we obtain
	\begin{align*}
		\frac{\partial }{\partial x_i} \log f_X(x) &= \frac{f'_{X_i}(x_i)}{f_{X_i}(x_i)}.
	\end{align*}
	Differentiating again with respect $x_k$ for $k \neq i$ gives 
	\begin{equation}
		\frac{\partial^2 }{\partial x_k \partial x_i} \log f_X(x) = 0, \: \forall k \neq i.
		\label{eqn: 2nd partial vanish}
	\end{equation}
	On the other hand, if condition (\ref{eqn: 2nd partial vanish}) is true, then $ \frac{\partial }{\partial X_i} \log f_X(x)$ is a function of $x_i$, and we can write $\log f_X(x)$ as the sum of a function on $x_i$ and a function on $\widehat{x_i}$, i.e.
	\begin{equation*}
		\log f_X(x) = g(x_i)+ h(\widehat{x_i})
	\end{equation*}
	for some function $g$ and $h$.  Then
	\begin{equation*} 
		f_X(x) = e^{g(x_i)} e^{h(\widehat{x_i})}
	\end{equation*} 
	Since $g$ and $h$ are unique up to a constant, we can normalize $g$, s.t. $\int e^{g(x_i)}dx_i = 1$.  Notice that also forces $\int e^{h(\widehat{x_i})} d\widehat x_i = 1$.  Therefore $e^{g(x_i)}$ and $e^{h(\widehat{x_i})}$ are the density functions of $X_i$ and $\widehat{x_i}$ respectively, and so 
	\begin{equation*}
		f_X(x) = f_{X_i}(x_i) f_{\widehat{X_i}}(\widehat{x_i}).
	\end{equation*}
	From this we conclude that $X_i$ is independent of $\widehat{X_i}$.
\end{proof}

\begin{lemma}\label{lemma: cross terms}
If $X$ is symmetric $\mathbb R^n$-valued random vector, then for any $1 \leq i, j \leq n$, $i \neq j$, the following holds:
\begin{align}
\mathbb E \left[ \frac{\frac{\partial f_X(X)}{\partial x_i}}{f_X(X)} \cdot \frac{\frac{\partial f_X(X)}{\partial x_j}}{f_X(X)} \right] = 0.
\end{align}
\end{lemma}
\begin{proof}[Proof of Lemma \ref{lemma: cross terms}]
Since $f_X(x_1, \dots, x_i, \dots, x_n) = f_X(x_1, \dots, -x_i, \dots, x_n)$, we have 
	\begin{align}
		\frac{\partial f_X}{\partial x_i}(x_1, \dots, -x_i, \dots, x_n) &= - \frac{\partial f_X}{\partial x_i}(x_1, \dots, x_i, \dots, x_n), \text{ and } \label{eqn_sym11}\\
		\frac{\partial f_X}{\partial x_j}(x_1, \dots, -x_i, \dots, x_n) &=  \frac{\partial f_X}{\partial x_j}(x_1, \dots, x_i, \dots, x_n).
		\label{eqn_sym21}
	\end{align}
	Therefore
	\begin{align*}
		\frac{\frac{\partial f_X}{\partial x_i}(x_1, \dots, -x_i, \dots, x_n)\cdot \frac{\partial f_X}{\partial x_j}(x_1, \dots, -x_i, \dots, x_n)}{f_X^2(x_1, \dots, -x_i, \dots, x_n)}\\
		=-\frac{\frac{\partial f_X}{\partial x_i}(x_1, \dots, x_i, \dots, x_n)\cdot \frac{\partial f_X}{\partial x_j}(x_1, \dots, x_i, \dots, x_n)}{f_X^2(x_1, \dots, x_i, \dots, x_n)}.
	\end{align*}
	Since the joint density is symmetric, we have 
	\begin{equation}
		\E \left[ \frac{ \frac{\partial f_X (X)}{\partial x_i} \cdot \frac{\partial f_X(X)}{\partial x_j}}{f_X(X)^2} \Bigg | X_1, \cdots, X_{i-1}, X_{i+1}, \cdots, X_n \right] = 0
	\end{equation}
	Taking an expectation again, we conclude
	\begin{align}
		\E \left [ \frac{ \frac{\partial f_X(X)}{\partial x_i}\cdot \frac{\partial f_X(X)}{\partial x_j}}{f_X(X)^2} \right] &= \E \left[ \E \left[ \frac{ \frac{\partial f_X(X)}{\partial x_i}\cdot \frac{\partial f_X(X)}{\partial x_j}}{f_X(X)^2} \Bigg | X_1, \dots, X_{i+1},X_{i+1}, \cdots, X_n\right]\right], \\
		&=0.
		\label{eqn_cond_exp1}
	\end{align}
	This concludes the proof.
\end{proof}

\begin{lemma}\label{lemma: independence}
	Let $X$ and $Y$ be real-valued random variables, and let $Z$ be an $\mathbb R^n$-valued random variable. We have the independence relations $X \perp \!\!\! \perp (Y,Z)$, and $X+Y  \perp \!\!\! \perp Z$. Furthermore, the characteristic function of $X$, denoted by $\phi_X$, is assumed to have no zeros. Then $X,Y,Z$ are mutually independent.
\end{lemma}
\begin{proof}
	Note that it is enough to show that $Y  \perp \!\!\!  \perp Z$. In terms of characteristic functions, we need to show that $\phi_{Y,Z}(t_1, t_2) = \phi_Y(t_1) \phi_Z(t_2)$ for all $t_1 \in \mathbb R$ and $t_2 \in \mathbb R^n$. Using the fact that $X$ is independent of $(Y,Z)$, we have
	\begin{align*}
	\mathbb E (e^{t_1X+t_1Y+t_2^T Z}) &= \phi_X(t_1) \phi_{Y,Z}(t_1, t_2).
	\end{align*}
	Using the independence of $(X+Y)$ and $Z$, we have
	\begin{align*}
	\mathbb E (e^{t_1X+t_1Y+t2^T Z}) &= \phi_{X+Y}(t_1) \phi_{Z}(t_2)\\
	&\stackrel{(a)}= \phi_X(t_1) \phi_Y(t_1) \phi_Z(t_2),
	\end{align*}
	where $(a)$ follows because $X \perp \!\!\! \perp (Y,Z)$. Since $\phi_X$ has no zeros, we may divide it out and conclude 
	\begin{align*}
	\phi_{Y,Z}(t_1, t_2) = \phi_Y(t_1) \phi_Z(t_2).
	\end{align*}
	This shows that $Y  \perp \!\!\! \perp Z$ and concludes the proof.
\end{proof}

\begin{lemma}\label{lemma: gaussian}
	If $X$ and $Y$ are real-valued random variables such that:
	\begin{enumerate}
		\item $X \sim \cN(0, \sigma_1^2)$
		\item $X  \perp \!\!\! \perp Y$
		\item$X+Y \sim \cN(0 , \sigma_1^2 + \sigma_2^2)$
	\end{enumerate}
	Then $Y \sim \cN(0, \sigma_2^2)$.
\end{lemma}
\begin{proof}
	Since $X$ and $Y$ are independent, we have 
	\begin{align*}
	\mathbb E(e^{tX + tY}) &= \phi_X(t) \phi_Y(t),
	\end{align*}
	and using the definition of $\phi_{X+Y}$, we have
	\begin{align*}
	\mathbb E(e^{tX + tY}) &= \phi_{X+Y}(t).
	\end{align*}
	Since $\phi_X$ has no zeros owing to $X$ being Gaussian, we have 
	\begin{align*}
	\phi_Y(t) = \frac{\phi_{X+Y}(t)}{\phi_X(t)}.
	\end{align*}
	The right hand side is precisely the characteristic function of a $\cN(0, \sigma_2^2)$ random variable. This implies $Y \sim \cN(0, \sigma_2^2)$ and concludes the proof.
\end{proof}

\end{appendix}

\bibstyle{ieeetr}
\bibliography{Ref}

\end{document}